\documentclass[
 reprint,
 amsmath,amssymb,
 aps,
 pra,
]{revtex4-2}

\usepackage{graphicx}
\usepackage{hyperref}
\usepackage[linesnumbered, ruled, vlined, titlenotnumbered, noend]{algorithm2e}
\usepackage{amsmath, amsfonts, amssymb, amsthm, color}
\usepackage{multirow} % For table in appendix
\usepackage{tikz-cd}
\usepackage[capitalise]{cleveref}
\crefname{algocf}{alg.}{algs.}
\Crefname{algocf}{Algorithm}{Algorithms}

\newtheorem{theorem}{Theorem}

\newtheorem{proposition}{Proposition}

\newtheorem{lemma}{Lemma}

\newcommand{\ket}[1]{|#1\rangle}

\newcommand{\Z}{{\mathbb{Z}}}

\newcommand{\F}{{\mathbb{F}}}

\newcommand{\HH}{{\bf{H}}}
\newcommand{\polyA}{{\cal A}}
\newcommand{\polyB}{{\cal B}}
\newcommand{\polyP}{{\cal P}}
\newcommand{\module}{M}
\newcommand{\datamodule}{M^{d}}
\newcommand{\ancillamodule}{M^{a}}

\begin{document}

\title{Distributed fault-tolerant quantum memories over a $2 \times L$ array of qubit modules}

\author{Edwin Tham}
\thanks{These authors contributed equally to this work.}
\author{Min Ye}
\thanks{These authors contributed equally to this work.}
\author{Ilia Khait}
\author{John Gamble}
\author{Nicolas Delfosse}
\affiliation{
    IonQ Inc.
}

\date{\today}

\begin{abstract}
We propose an architecture for a quantum memory distributed over a $2 \times L$ array of modules equipped with a cyclic shift implemented via flying qubits.
The logical information is distributed across the first row of $L$ modules and quantum error correction is executed using ancilla modules on the second row equipped with a cyclic shift.
This work proves that quantum LDPC codes such as BB codes can maintain their performance in a distributed setting while using solely one simple connector: a cyclic shift.
We propose two strategies to perform quantum error correction on a $2 \times L$ module array:
(i) The cyclic layout which applies to any stabilizer codes, whereas previous results for qubit arrays are limited to CSS codes.
(ii) The sparse cyclic layout, specific to bivariate bicycle (BB) codes.
For the $[[144,12,12]]$ BB code, using the sparse cyclic layout we obtain a quantum memory with $12$ logical qubits distributed over $12$ modules, containing $12$ physical qubits each.
We propose physical implementations of this architecture using flying qubits, that can be faithfully transported, and include qubits encoded in ions, neutral atoms, electrons or photons.
We performed numerical simulations when modules are long ion chains and when modules are single-qubit arrays of ions showing that the distributed BB code achieves a logical error rate below $2 \cdot 10^{-6}$ when the physical error rate is $10^{-3}$.
\end{abstract}

\maketitle

\section{Introduction}

Large-scale quantum applications might require millions of physical qubits, due to the large overhead of quantum error correction and fault-tolerance~\cite{reiher2017elucidating, beverland2022assessing, dalzell2023quantum, gidney2025factor, zhou2025resource}.
Modular designs are appealing because they simplify the manufacture, testing and characterization of large-scale quantum chips.
Modularity circumvents issues specific to certain implementation modalities as well: the spectral crowding of collective motional modes for trapped ions~\cite{landsman2019two, shapira2023fast}, the dropping yield of superconducting chips~\cite{zeissler2024superconducting, ang2024arquin}, the laser power limitation of neutral atoms~\cite{henriet2020quantum}, and cryogenic scaling requirements~\cite{fellous2023optimizing}.
However, it also gives rise to two critical challenges: connecting the modules and designing a distributed architecture for fault-tolerant quantum computing.

A popular approach to distributed quantum computing is based on small modules connected through noisy links~\cite{jiang2007distributed, li2012high, fujii2012distributed, monroe2014large}, with entanglement distillation~\cite{bennett1996purification} used to extract high-fidelity gates from these links.
Work on distributed surface codes~\cite{ramette2024fault, de2024thresholds} and Floquet codes~\cite{sutcliffe2025distributed} shows that these codes perform well even when a small fraction of the gates are implemented through very noisy links without distillation.
However, this approach seems challenging for general quantum low-density parity-check (LDPC) codes~\cite{breuckmann2021quantum} because the qubit connectivity they require is typically an expander graph, meaning that it cannot be easily partitioned into modules with few connections between the modules~\cite{hoory2006expander}.

In the present work, the module connection is established by physically moving the qubits.
The ability to reliably transport qubits was identified by DiVicenzo as an essential requirement for qubits used for quantum computation and communication and he named them {\em flying qubits}~\cite{divincenzo2000physical}.
They include photonic qubits~\cite{knill2001scheme}, spin qubits~\cite{loss1998quantum}, electron on liquid helium~\cite{lyon2006spin}, trapped ions~\cite{cirac1995quantum} and neutral atoms~\cite{bluvstein2024logical}.
Here, we propose a distributed quantum error correction scheme supported on a $2 \times L$ array of modules connected through a cyclic shift of the modules implemented using flying qubits.
We simulated the performance of distributed bivariate bicycle codes~\cite{bravyi2024high} for this architecture where the modules are with long ion chains and when modules are one-dimensional arrays of ions.
The results show that our modular quantum memory can reach the low logical error rates required for large-scale applications.

In the remainder of this paper, \cref{sec:The 2xL model} proposes an abstract model for a $2 \times L$ modules array. \cref{sec:The cyclic layout} introduces the cyclic layout which allows for the implementation of the syndrome extraction circuit of any stabilizer code. A sparse cyclic layout, producing a constant-depth syndrome extraction circuit for BB codes is proposed in \cref{sec:The sparse cyclic layout}.
Potential physical implementations and numerical simulations are discussed in \cref{sec:physical implementation} and \cref{appendix:Flat implementation}.

\section{The $2 \times L$ model}
\label{sec:The 2xL model}

We consider a $2 \times L$ array equipped with a cyclic shift where each cell may contain a register of qubits that we call a {\em module}.
This generalizes the $2 \times L$ array of qubits introduced in~\cite{siegel2024towards}. 
We refer to this generalization as a {\em $2 \times L$ module array} and we use the term {\em $2 \times L$ qubit array} for the original model which corresponds to single-qubit modules.

The cells of the array are labeled $(b, i) \in \Z_2 \times \Z_L$, where cells $(0, i), i\in\Z_L$ form the {\em fixed row} and cells $(1, i), i\in\Z_L$ form the {\em moving row}.
Each cell is either empty or it contains an $n$-qubit module.
For simplicity, we assume that all the modules are identical $n$-qubit registers.

The {\em qubit operations} available are
preparation or reset of a qubit in a single-qubit state,
measurement of a qubit,
single-qubit unitary gates,
and two-qubit unitary gates supported inside a module or in a pair of {\em aligned modules}, that is acting on qubits in cells $(0,i)$ and $(1, i)$.

A {\em cyclic shift} with size $s$, or {\em $s$-shift}, moves all the modules of the moving row by $s$ steps to the right in a cyclic way, where $s$ is any integer.
The module in cell $i$ of the moving row is transported to the cell $(i + s) \mod L$.

We assume that operations acting on different cells can be performed simultaneously.
Moreover, the measurement and a reset of a qubit can be performed in a single step.
Any cyclic shift has depth one, independently of the shift size $s$. 
That is, the shift operation duration is independent of the physical distance of the shift.
Depending on the details of the gate operations, transport speeds, and other modality-specific physical details, this assumption may break down.
We discuss its validity further in ~\cref{sec:physical implementation}.

In \cref{sec:The cyclic layout,,sec:The sparse cyclic layout}, we assume that each module or pair of aligned modules forms a {\em fully connected and fully parallel} qubit register, meaning that any set of two-qubit gates with disjoint supports can be executed in depth one.
We study other cases in \cref{sec:physical implementation} and \cref{appendix:Flat implementation}.

\section{The cyclic layout}
\label{sec:The cyclic layout}

The {\em cyclic layout}, described in \cref{algorithm:cyclic_layout}, performs the measurement of any sequence of Pauli operators on a $2 \times L$ module array.
It implements the syndrome extraction circuit of any stabilizer code by providing as an input the code's stabilizer generators (repeated $T$ times to perform $T$ rounds of syndrome extraction).

\begin{algorithm}
\DontPrintSemicolon
\SetAlgoLined
\KwIn{
A $2 \times L$ module array.
A list of $N$-qubit Pauli operators $P_0, P_1, \dots P_{r-1}$ supported on the first $L-1$ cells of the fixed row.}
\KwOut{A quantum circuit measuring the input Pauli operators over the $2 \times L$ module array.}
    Assign the identity operator $I$ to all the qubits of the moving row and define $P_t := I$ for all $t > r-1$.\;
    Let $\module$ be the last module of the moving row.\;    
    Prepare all the qubit of $\module$ in the $\ket +$ state.\;
    Assign the $P_0, P_1, \dots, P_{n-1}$ to the qubits of $\module$ and mark them.\;
    \For{$t=1,2,\dots, \left \lceil r/n \right \rceil + L$}{
        Apply a $1$-shift.\;
        \For{all module $\module$ on the first $L-1$ cells of the moving row}{
            Let $\module'$ be the module aligned with $\module$.\;
            \For{qubit $i$ in $\module$ and qubit $j$ in $\module'$}{
                If the operator assigned to qubit $i$ acts as $Q_j \neq I$ on qubit $j$, apply a controlled-$Q_j$ gate controlled on qubit $i$ targeting qubit~$j$.
            }
        }
        Let $\module$ be the last module of the moving row.\;
        Measure and reset all the qubits of $\module$ in the $X$ basis.\;
        Assign the first $n$ unmarked operators $P_i$ to the qubits of $\module$ and mark them.\;
    }
\caption{Cyclic layout for stabilizer codes.}
\label{algorithm:cyclic_layout}
\end{algorithm}

Consider an $N$-qubit Pauli operator $Q = Q_1 \times \dots \otimes Q_N$ where $Q_j$ is a Pauli matrix and refer to the $N$ qubits supporting $Q$ as the data qubits.
One can perform the measurement of $Q$ in three steps as follows:
(i) prepare an ancilla qubit in the state $\ket +$,
(ii) apply a sequence of controlled-$Q_j$ gates controlled on the ancilla qubit and targeting the $j$ th data qubit for $1\le j\le N$,
(iii) measure the ancilla qubit in the $X$ basis.

\cref{algorithm:cyclic_layout} measures simultaneously Pauli operators supported on the fixed row of a $2 \times L$ module array using ancilla qubits placed on the moving row.
The main challenge is to design a sequence of cyclic shifts that allows for the implementation of the two-qubit gates required for the measurement of all the Pauli operators without swapping gates associated to different operators because these gates generally do not commute.
To obtain this property, the loop of step 10 is always executed in the same order.

\begin{proposition}\label{proposition:cyclicstab}
\cref{algorithm:cyclic_layout} performs the measurement of $r$ $N$-qubit Pauli operators on a $2 \times L$ array of $n$-qubit modules in depth at most 
$
3 + (\lceil r/n \rceil + L-1)(n+1)
$.
\end{proposition}

\begin{proof}
Consider two operators $P_t$ and $P_{t'}$ with $t < t'$.
If $P_t$ and $P_{t'}$ are assigned to two ancilla qubits of the same module $\module$,
then all the controlled-Pauli gates associated with $P_t$ are executed before the controlled-Pauli gates associated with $P_{t'}$ in step 10.
Assume now that $P_t$ and $P_{t'}$ are assigned to ancilla qubits in different modules $\module_t$ and $\module_{t'}$, where $\module_{t}$ is reset before $\module_{t'}$.
Again, the controlled-Pauli gates controlled on $\module_t$ targeting a given module are performed before the gates controlled on $\module_{t'}$ targeting the same module.
This proves that the circuit is equivalent to the sequential measurement of the Pauli operators.

After the first preparation, for 
$t=1,2,\dots, \lceil r/n \rceil + L - 1$,
we perform a cyclic shift, a sequence of two-qubit gates acting on the first $L-1$ pairs of aligned modules, and a measurement and reset on the last cell of the moving row. 
The two-qubit gates can be implemented in depth at most $n$ and the measurement and reset can be performed at the same time.
For the last value of $t$, there are no more two-qubit gates to execute.
This yields the upper bound 
$
1 + (\lceil r/n \rceil + L-1)(n+1) + 2
$
on the depth.
\end{proof}

The main advantage of \cref{algorithm:cyclic_layout} is that it applies to any stabilizer code.
It is practically relevant for small codes. 
However, when the number of stabilizer generators $s \rightarrow + \infty$, the syndrome extraction depth becomes too large, degrading the code performance.
Indeed, the bound on the depth per round tends to
$
\frac{(n+1)}{n}s
$.

\section{The sparse cyclic layout}
\label{sec:The sparse cyclic layout}

The sparse cyclic layout produces a short-depth syndrome extraction circuit for BB codes~\cite{bravyi2024high, kovalev2013quantum}.

Denote by $S_\ell$ the $\ell \times \ell$ circulant matrix with first row $(0 1 0 \dots 0)$ and let 
$x = S_\ell\otimes I_m$ and $y = I_{\ell}\otimes S_m$.
The BB code associated with the polynomials 
$\polyA, \polyB \in \F_2[x,y]$ is defined to be the CSS code~\cite{calderbank1996good, steane1996multiple} with parity-check matrices $\HH_X=[\polyA | \polyB]$ and $\HH_Z = [\polyB^T | \polyA^T]$.
Therein, $\polyA$ and $\polyB$ are sums of matrices of the form $x^i y^j$.
In~\cite{bravyi2024high}, these polynomials are constrained to have exactly three terms, and each term is a power of either $x$ or $y$. 
Here, we allow for any polynomial, which allows one to reach better code parameters~\cite{voss2024multivariate, ye2025quantum}.

Given a polynomial
$
\polyP(x,y) = x^{i_1}y^{j_1} + \dots + x^{i_t}y^{j_t}
$
define $I(\polyP) := \{i_1, i_2, \dots, i_t\}$ and $J(\polyP) := \{j_1, j_2, \dots, j_t\}$ to be the set of distinct exponents of $x$ and $y$ in $\polyP$. 
Based on $\left(x^i z^j\right)^T = x^{-i}y^{-j}$, the set $I(\polyP^T)$ and $J(\polyP^T)$ are obtained by replacing the elements of $I(\polyP)$ and $J(\polyP)$ by their opposite.

Any $k \in \{0,1,\dots \ell m - 1\}$ can be mapped onto the element $\left( \lfloor k / m \rfloor, k \mod m \right)$ of $G_{\ell, m} := \Z_\ell \times \Z_m$.
This bijection allows us to label rows and columns of a matrix $x^i y^j$ with elements of $G_{\ell, m}$.
Examining the matrix $x^i y^j$, we obtain the following lemma where $\oplus$ denotes the addition modulo $\ell$ or modulo $m$. The modulus is clear from the context.

\begin{lemma}
\label{lemma:coefficient_of_xi_yj}
The coefficient of the matrix $x^i y^j$ in row $(v,w) \in G_{\ell, m}$ and column $(v',w') \in G_{\ell, m}$ is 1 iff $(v',w') = (v \oplus i,w \oplus j)$.
\end{lemma}

Extending the previous bijection, we label the code's data qubits with $G_{2,\ell, m} := \Z_2 \times \Z_\ell \times \Z_m$.
The triple $(u,v,w)$ corresponds to the data qubit with index $u \ell m + v m + w$.
The ancilla qubits, which correspond to the rows of $\HH_X$ and $\HH_Z$, are labeled respectively as $(X,v,w)$ and $(Z,v,w)$ with $(v,w) \in G_{\ell,m}$.

Define the {\em data modules}
$
\datamodule_{w} := \Z_2 \times \Z_\ell  \times \{w\}
$
indexed by $w \in \Z_m$, which we interpret as sets of data qubits. 
Define the {\em ancilla modules} 
$
\ancillamodule_{w} := \{X,Z\} \times \Z_\ell  \times \{w\},
$
also indexed by $w \in \Z_m$.

These modules form a $2 \times m$ array with $2 \ell$-qubit modules.
Modules $\datamodule_{w}$ and $\ancillamodule_{w}$ are initially placed in cell $w$ of the fixed row and the moving row respectively.

\begin{algorithm}
\DontPrintSemicolon
\SetAlgoLined
\KwIn{A BB code.\;}
\KwOut{A circuit measuring the $X$ stabilizer generators of the input code over the $2 \times m$ module array.\;}
Prepare all the $X$ ancilla qubits in the state $\ket +$.\;
\For{$j \in J(\polyA) \cup J(\polyB)$}{
    Apply the cyclic shift aligning $\ancillamodule_{0}$ and $\datamodule_{j}$.\;
    \For{$i \in \Z_m$ such that $x^i y^j$ appears in $\polyA$}{
        \For{$v, w \in G_{\ell, m}$}{
        Apply the CX gate controlled on qubit $(X,v,w)$ targeting qubit $(0, v \oplus i, w \oplus j)$.\;
        }
    }
    \For{$i \in \Z_m$ such that $x^i y^j$ appears in $\polyB$}{
        \For{$v, w \in G_{\ell, m}$}{
        Apply the CX gate controlled on qubit $(X, v,w)$ targeting qubit $(1,v \oplus i, w \oplus j)$.\;
        }
    }
}
Measure all the ancilla qubits in the $X$ basis.\;
\caption{Sparse cyclic layout for BB codes.}
\label{algorithm:sparse_cyclic_layout}
\end{algorithm}

\begin{proposition}
\cref{algorithm:sparse_cyclic_layout} performs the measurement of the $X$ stabilizer generators of the input BB code.
\end{proposition}

\begin{proof}
The CX gates implemented at steps 6 and 9 are valid because the cyclic shift at step 3 aligns modules $\ancillamodule_{w}$ and $\datamodule_{w \oplus j}$ supporting these gates. 
This is because the sum $w \oplus j$ is taken modulo $m$ which coincides with the period of the cyclic shift.

Based on \cref{lemma:coefficient_of_xi_yj}, to measure the $X$ stabilizer generator associated with row $(v,w)$ of $\HH_X$, we need to perform CX gates controlled on qubit $(X,v,w)$ targeting qubit $(0,v \oplus i,w \oplus j)$ for each term $x^iy^j$ in $\polyA$ and $(1, v \oplus i, w \oplus j)$ for each term $x^iy^j$ in $\polyB$.
Theses gates are implemented in steps 6 and 9 of \cref{algorithm:sparse_cyclic_layout}. 
\end{proof}

\cref{algorithm:sparse_cyclic_layout} only describes $X$ stabilizer measurements because $Z$ measurements can be performed similarly.

\begin{theorem}
\label{theorem:constant_depth_BB_syndrome_extraction}
\cref{algorithm:sparse_cyclic_layout} performs the $X$ syndrome extraction of a BB code in depth 
$|J(\polyA) \cup J(\polyB)| + \omega + 2$ 
using a $2 \times m$ module array
where $\omega$ is the weight of the stabilizer generators.
The same holds for the $Z$ syndrome extraction.
\end{theorem}

By symmetry one can swap the roles of $\ell$ and $m$ in \cref{algorithm:sparse_cyclic_layout}. In this case, the depth in \cref{theorem:constant_depth_BB_syndrome_extraction} becomes $|I(\polyA) \cup I(\polyB)| + \omega + 2$, which may be smaller than $|J(\polyA) \cup J(\polyB)| + \omega + 2$.

Applying \cref{theorem:constant_depth_BB_syndrome_extraction}, we obtain an $X$ or $Z$ syndrome extraction circuit with depth $12$ for all the BB codes of~\cite{bravyi2024high}.

\begin{proof}
The first and last instructions account for two steps and there are a total of $|J(\polyA) \cup J(\polyB)|$ cyclic shifts.
Inside the loop of step 5, we perform $\ell m$ CX gates
which can be implemented simultaneously because they act on disjoint pairs of qubits.
Similarly, the $\ell m$ CX gates in the loop of step 8 can be implemented in depth one.
Therefore, the measurement of the $\ell m$ $X$ stabilizer generators, which requires a total of $\omega \ell m$ CX gates, can be performed in depth $\omega$.

The $Z$ stabilizer measurements are performed similarly based on the transposed matrices $\polyB^T$ and $\polyA^T$. Given that $|J(\polyA^T) \cup J(\polyB^T)|=|J(\polyA) \cup J(\polyB)|$, the $Z$ measurement depth is the same.
\end{proof}

\cref{appendix:interleaved implementation} discusses a variant of \cref{algorithm:sparse_cyclic_layout} with interleaved $X$ and $Z$ measurements achieving a shorter depth.

\section{Physical implementation}
\label{sec:physical implementation}

\begin{figure}[htbp]
    \centering 
    
    \includegraphics[width=.95\linewidth]{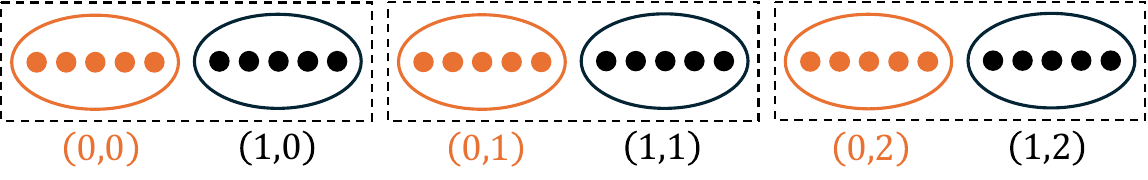}

    \caption{Representation of a $2 \times 3$ module array using a 1D line of qubits with five qubits per module. Modules of the fixed row and moving row are alternating and aligned modules are inside the dashed boxes.}
    \label{fig:one_dime_2xL_module_array}
\end{figure}

Here, we describe a quasi one-dimensional implementation of a $2 \times L$ module array and its cyclic shift using flying qubits which could be photons, electrons, ions or neutral atoms.

\begin{figure}[htbp]
    \centering
    \fbox{
    \includegraphics[width=.9\linewidth]{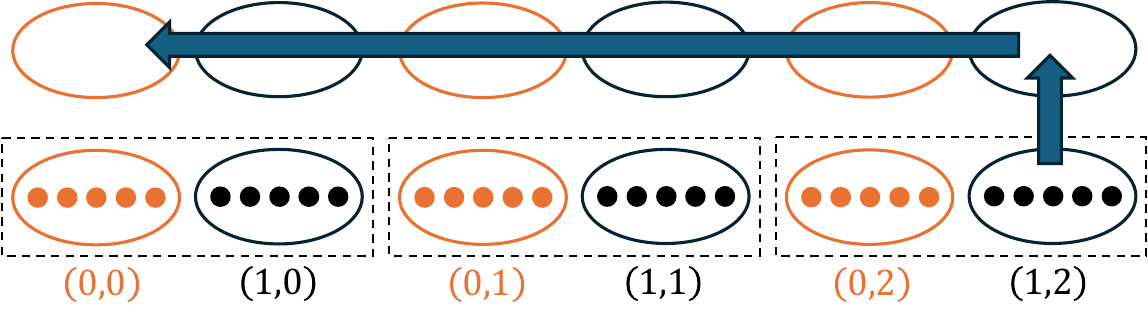}
    }

    (a)

    \fbox{
    \includegraphics[width=.9\linewidth]{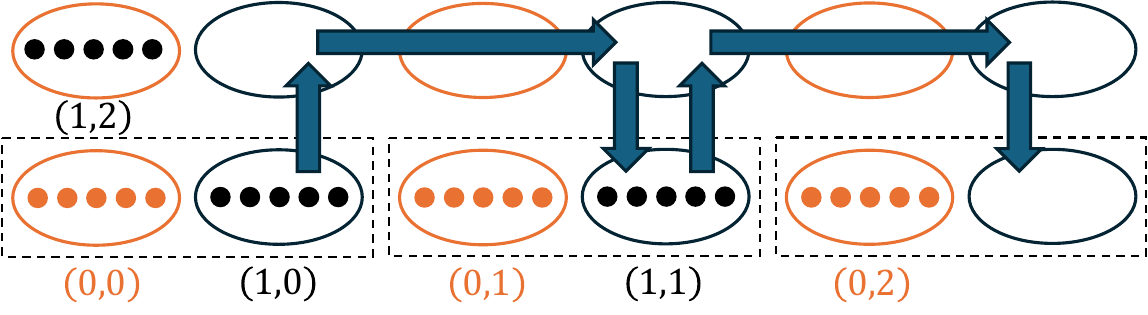}
    }

    (b)

    \fbox{
    \includegraphics[width=.9\linewidth]{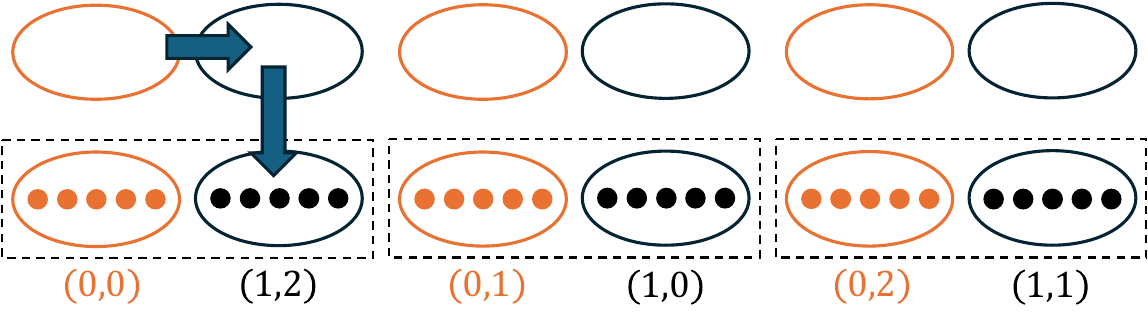}
    }

    (c)

    \caption{Representation of a 1-shift on a $2 \times 3$ module array in three steps using a primary zone (bottom) storing all the modules and a secondary zone for temporary storage.
    }
    \label{fig:phys_imp_cyclic_shift}
\end{figure}

A $2 \times L$ array with $n$-qubit modules is formed using $2Ln$ flying qubits arranged within a line and split into groups of $n$ qubits as shown in \cref{fig:one_dime_2xL_module_array}.
The modules are alternating between modules of the fixed row and the moving row.
We assume that one can perform two-qubit gates between neighboring modules as shown by the dashed boxes in \cref{fig:one_dime_2xL_module_array}.
In practice, this might require bringing the qubits of these two modules closer together, which is not an issue for flying qubits.

We refer to the one-dimensional region holding the modules as the {\em primary zone}.
To facilitate cyclic shifts, we use a parallel {\em secondary zone}, represented in \cref{fig:phys_imp_cyclic_shift}, for temporary storage of the modules.

A $s$-shift is realized in three moves as illustrated in \cref{fig:phys_imp_cyclic_shift}:
(i) Move the last $s$ modules of the moving row to the secondary zone.
(ii) Move the first $L-s$ modules of the moving row forward by $2s$ positions in the primary zone.
(iii) Move the modules present in the secondary zone to the first $s$ cells of the moving row in the primary zone.

Step (ii) is accomplished by moving the relevant modules into the secondary zone, moving them forward, and returning them to the primary zone.
This is more efficient than swapping qubits which requires a number of swaps growing with $s$.
To accommodate size-$s$ shifts, the secondary zone must be physically longer by an extra $s-1$ module cells than what is strictly needed to hold a length-$L$ module array.

Within any of these three steps, all the moves can be performed simultaneously while keeping the modules sufficiently far from each other to avoid unwanted interactions.
Even though qubits must physically be transported across distances proportional to $s$, in practice for modest distances spanning hundreds of $\mu$m, overall transport times and noise remain dominated by fixed-duration processes that are independent of $s$, such as acceleration, deceleration and cooling in the case of ions.
This justifies our assumption that any cyclic shift is implemented in depth one, independently of the shift size $s$.

To assess the performance of this architecture, we performed circuit-level simulations of BB codes with \cref{algorithm:sparse_cyclic_layout} using this implementation of the $2 \times L$ module array and its cyclic shift where each module is a long chain of trapped ions.
We used the chain model of~\cite{ye2025quantum} to simulate qubit operations inside modules.
Two-qubit gates are sequential inside a module but gates acting on distinct modules can be performed simultaneously.
Two-qubit gates have a noise rate $p$, single-qubit operations have a noise rate $p/10$, and idle qubits have a noise rate $p/100$.
We assume $\tau_m = 30$, meaning that unmeasured qubits undergo $30$ rounds of idle noise during a measurement.
Finally, a cyclic shift is followed by depolarizing noise on all the qubits with rate $\tau_s p/100$ with $\tau_s = 30$, which means that all qubits suffer from $\tau_s$ rounds of idle noise.

\cref{fig:simulation_results_long_chain} shows that the BB code with length 144 distributed across 12 ion-chain modules achieves a logical error rate below $2 \cdot 10^{-6}$ for a physical error rate of $10^{-3}$.
A different implementation based on flat modules which are one-dimensional array of qubits is proposed and simulated in \cref{appendix:Flat implementation}.

In \cref{appendix:Impact of modularity}, we analyze the impact of distributing the codes over several modules on the code performance and we observe that it comparable to increase on the physical error rate $p$ by less than $2 \times$.
In \cref{appendix:fitting_formulas}, we provide a fitting formula for the logical error rate of BB codes under the sparse cyclic layout.

\begin{figure}[htbp]
    \centering
    \includegraphics[width=.95\linewidth]{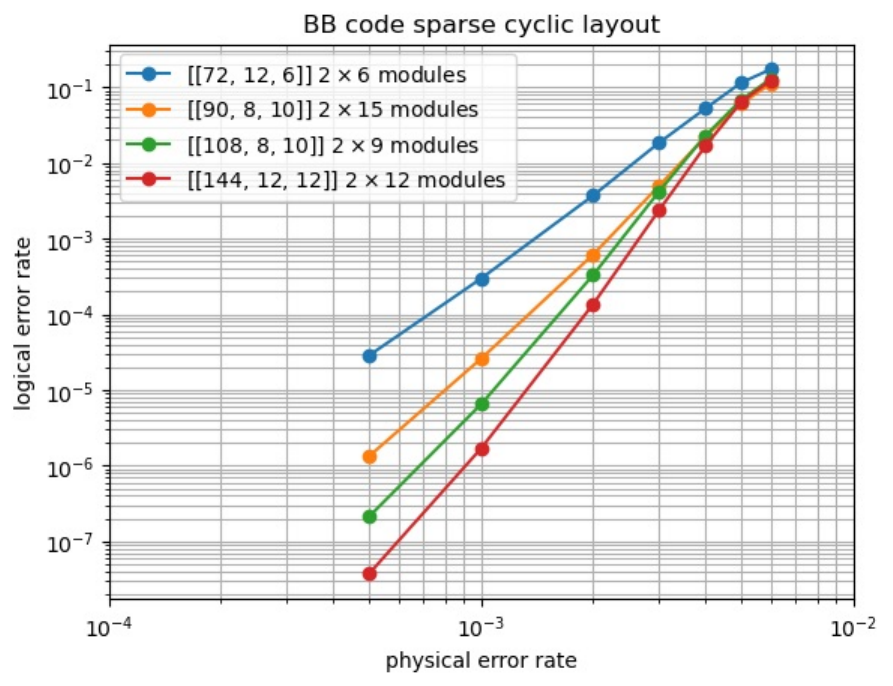}
    \caption{Logical error rate of BB codes with the sparse cyclic layout of \cref{algorithm:sparse_cyclic_layout} distributed over a module array where each module is a long chain of trapped ions.}
    \label{fig:simulation_results_long_chain}
\end{figure}

\section{Conclusion}
\label{sec:conclusion}

We proposed a design for a distributed quantum memory implemented with flying qubits.
Although we use DiVicenzo's concept of flying qubits, our architecture only requires a planar motion of the qubits, which we may call {\em movable qubits}, making it well-suited to electrons, ions and neutral atoms.
It would be valuable these notions of transports to distinguish different types of flying qubits such as ions, atoms, electrons, photons or even qubits loaded on a cargo ship~\cite{devitt2016high} and to identify more precise requirements for our architecture.

It would be interesting to generalize this layout to other classes of quantum codes. 
For qubit arrays, \cite{siegel2024towards} layouts surface codes and generalized bicycle codes. 
The flat implementation of the sparse cyclic layout, discussed in \cref{appendix:Flat implementation}, is related to this generalized bicycle layout, with the difference that they use left and right moves instead of a cyclic shift.
An alternative implementation of BB codes in a qubit array, relying on Shor-style error correction which consumes more ancilla qubits, is optimized in \cite{micciche2025optimizing}.

As explained in the introduction, quantum LDPC codes are generally hard to partition because of their expansion.
A related result is the following.
Using a finite dimensional grid of qubits with local gates without the cyclic shift, a constant depth syndrome extraction circuit cannot exist if the code's Tanner graph is locally expanding~\cite{delfosse2021bounds}.
Graph expansion is also used to establish bounds on LDPC codes' parameters~\cite{baspin2022connectivity}.
It would be interesting to understand the impact of the ability to perform a cyclic shift, and more generally the impact of flying qubits, on these bounds and other bounds on codes and logical operations~\cite{bravyi2010tradeoffs, bravyi2013classification, pastawski2015fault, jochym2018disjointness}.

\section{Acknowledgment}

We thank Jeremy Sage, Dave Wecker, Matthew Parrott, Jason Amini for their insightful discussions and for their comments on a preliminary version of this work.

\appendix

\section{Flat implementation}
\label{appendix:Flat implementation}

\begin{algorithm}
\DontPrintSemicolon
\SetAlgoLined
\KwIn{A BB code with code length $N$.\;}
\KwOut{A circuit measuring the $X$ stabilizer generators of the input code over the $2 \times N$ qubit array.\;}
Prepare all the $X$ ancilla qubits in the state $\ket +$.\;
\For{$j \in J(\polyA)$}{
    Apply the cyclic shift aligning $\ancillamodule_{0}$ and $\datamodule_{j}$.\;
    \For{$i \in \Z_m$ such that $x^i y^j$ appears in $\polyA$}{
        \For{$w \in \Z_m$}{
            Apply the intra-module cyclic shift (with period $2\ell$) in module $\ancillamodule_{w}$ aligning qubit $(X,0,w)$ with qubit $(0, i, w \oplus j)$.
        }
        \For{$v, w \in G_{\ell, m}$}{
        Apply the CX gate controlled on qubit $(X,v,w)$ targeting qubit $(0, v \oplus i, w \oplus j)$.\;
        }
    }
}
\For{$j \in J(\polyB)$}{
    Apply the cyclic shift aligning $\ancillamodule_{0}$ and $\datamodule_{j}$.\;
    \For{$i \in \Z_m$ such that $x^i y^j$ appears in $\polyB$}{
        \For{$w \in \Z_m$}{
            Apply the intra-module cyclic shift (with period $2\ell$) in module $\ancillamodule_{w}$ aligning qubit $(X,0,w)$ with qubit $(1, i, w \oplus j)$.
        }
        \For{$v, w \in G_{\ell, m}$}{
        Apply the CX gate controlled on qubit $(X,v,w)$ targeting qubit $(1, v \oplus i, w \oplus j)$.\;
        }
    }
}
Measure all the ancilla qubits in the $X$ basis.\;
\caption{Flat cyclic layout for BB codes.}
\label{algorithm:flat_layout}
\end{algorithm}

\begin{figure}[htbp]
    \centering
    \includegraphics[width=0.95\linewidth]{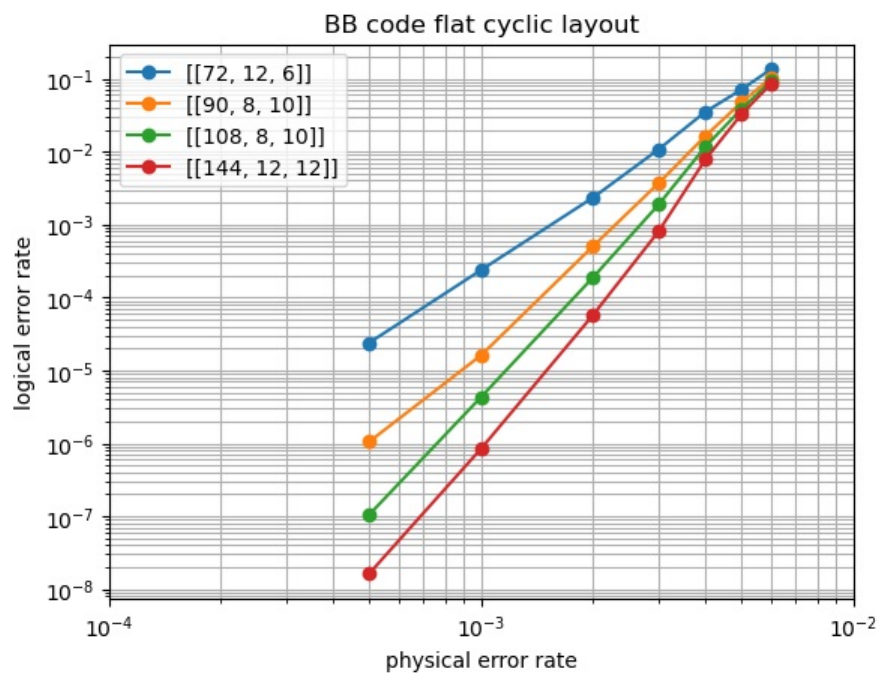}
    \caption{Logical error rate of BB codes with the flat cyclic layout of \cref{algorithm:flat_layout}.}
    \label{fig:simulation_results_flat}
\end{figure}

An alternative to long chains is to implement each $n$-qubit module as a one-dimensional array of $n$ qubits.
When two such modules are aligned, the CX gates on the $n$ aligned pairs of qubits can be executed simultaneously.
Moreover, we assume that an intra-module cyclic shift with period $n$ is available as shown in \cref{fig:flat_2x3_module_array}.
Each of these modules can be built with the approach described in \cref{sec:physical implementation} using flying qubits.

\begin{figure}[htbp]
    \centering
    \includegraphics[width=.9\linewidth]{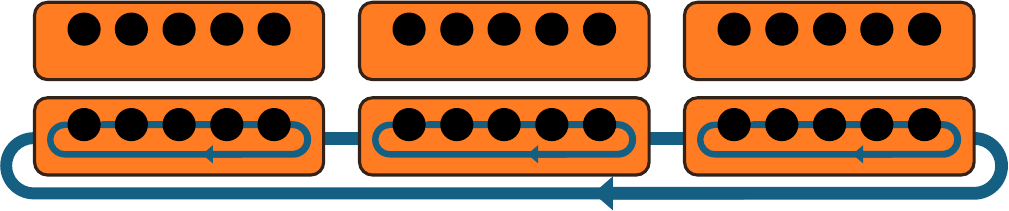}
    \caption{A flat $2 \times 3$ module array with 5-qubit modules equipped with a cyclic shift of the modules and intra-module cyclic shifts.}
    \label{fig:flat_2x3_module_array}
\end{figure}

To implement \cref{algorithm:sparse_cyclic_layout} with such flat modules, we set $n=2\ell$, and the qubits of $\datamodule_{w}$ are placed in a one-dimensional array in the following order
\begin{align*}
    (0,0,w), (1,0,w), (0,1,w), (1,1,w), \dots, (1,\ell-1, w)
\end{align*}
alternating between left and right data qubits.
The ancilla modules $\ancillamodule_{w}$ are built similarly, aligning $(X,v,w)$ with $(0,v,w)$ and $(Z,v,w)$ with $(1,v,w)$.

To execute \cref{algorithm:sparse_cyclic_layout} with flat modules, a round of intra-module cyclic shifts must be inserted before each round of CX gates, resulting in a flat implementation of \cref{algorithm:sparse_cyclic_layout} over a $2 \times N$ qubit array equipped with a global cyclic shift with period $N$ and intra-module cyclic shifts with period $2\ell$.

The flat implementation, whose pseudo-code is provided in \cref{algorithm:flat_layout}, uses more cyclic shifts (up to two for each monomial) than the long chain implementation but fewer rounds of CX gates because the CX gates associated with each monomial can be implemented simultaneously.
Precisely, the depth of the $X$ syndrome extraction circuit is at most 
$
|J(\polyA) \cup J(\polyB)| + \omega + 2
$
in the long chain case and at most
$
3\omega + 2
$
in the flat case.

The performance of BB codes with the syndrome extraction circuit of \cref{algorithm:flat_layout} when modules are one-dimensional arrays of ions is shown in \cref{fig:simulation_results_flat}. 
The simulation uses the ion chain model of~\cite{ye2025quantum} with single-qubit chains (merged into two-qubit chain for the duration of a two-qubit gate). We use $\tau_m = 30$ and $\tau_s = 10$ to simulate noisy operations.
We set $\tau_s=10$ here as opposed to $\tau_s=30$ for the long ion-chain module in \cref{sec:physical implementation} to reflect the faster transport of single-qubit ion chains.
We observe that BB codes exhibit slightly better performance under the flat cyclic layout than the sparse cyclic layout in \cref{fig:simulation_results_long_chain}.
In \cref{appendix:fitting_formulas}, we provide a fitting formula for the logical error rate of BB codes under the flat cyclic layout.

For convenience, we described the flat layout in terms of the cyclic shifts used throughout this paper.
However, these cyclic shifts could be replaced by left and right moves of the moving row, resulting in a properly one-dimensional implementation the flat cyclic layout.

\section{Impact of modularity}
\label{appendix:Impact of modularity}

\begin{figure}[htbp]
    \centering
    \includegraphics[width=0.95\linewidth]{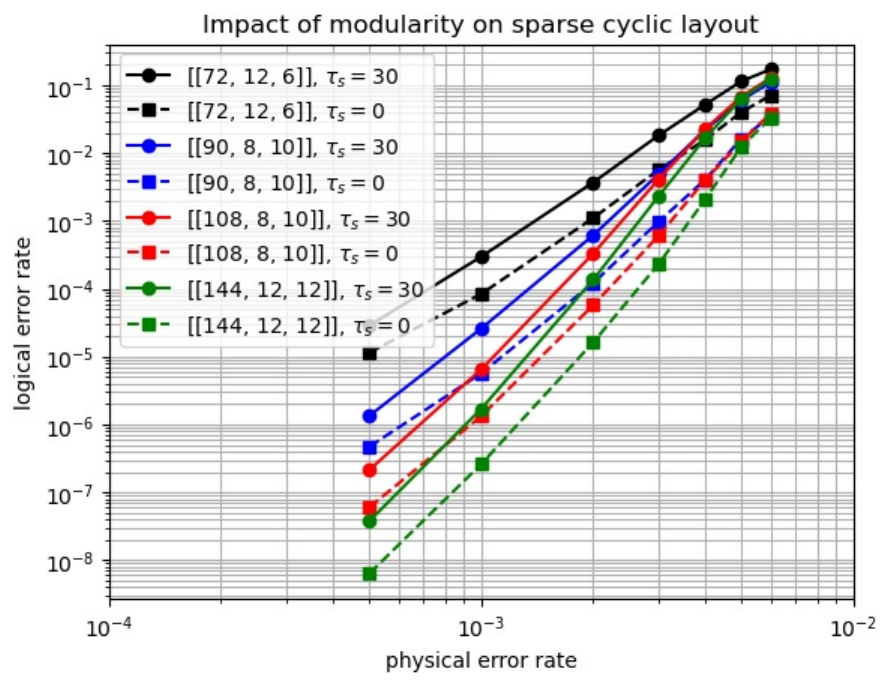}
    \caption{The 4 solid-line curves are exactly the same as the 4 curves in \cref{fig:simulation_results_long_chain}. The 4 dashed-line curves are obtained by setting $\tau_s=0$, which means that there is no noise associated with cyclic shifts.}
    \label{fig:simulation_results_modularity}
\end{figure}

The modular, or distributed, nature of our model is reflected by the necessity of aligning different modules with cyclic shifts in order to apply two-qubit gates across them. In circuit-level simulations of \cref{fig:simulation_results_long_chain}, these shifts are assumed to induce $\tau_s=30$ rounds of idle noise on all qubits, consequently increasing the logical error rate of the quantum error correction scheme.
To measure the impact of modularity on the code performance, we simulate the BB codes in the same setting as in \cref{fig:simulation_results_long_chain} but with $\tau_s = 0$, making the cyclic shifts noiseless. The performance comparison between $\tau_s=30$ and $\tau_s=0$ is given in \cref{fig:simulation_results_modularity} for 4 BB code instances. 
We use $p_{\log}(p, \tau_s=30)$ to denote the logical error rate on the curve $\tau_s=30$ at physical error rate $p$, and we define $p_{\log}(p, \tau_s=0)$ in a similar way. It is clear from \cref{fig:simulation_results_modularity} that $p_{\log}(p, \tau_s=30) < p_{\log}(2p, \tau_s=0)$ for all physical error rate $p$ and all 4 BB code instances. Therefore, in order for the noisy cyclic shift model to achieve the same logical error rate as the noiseless cyclic shift model, we only need to decrease the physical error rate by a factor of at most 2. In other words, the impact of modularity is a factor of at most 2 on the physical error rate.

\section{Interleaved version of the sparse cyclic layout}
\label{appendix:interleaved implementation}

\cref{sec:The sparse cyclic layout} describes the measurement of the $X$ stabilizer generators of the BB codes.
Applying \cref{algorithm:sparse_cyclic_layout} twice -- once for each stabilizer type -- is sufficient to fully implement the syndrome extraction circuit.
This section describes variants of \cref{algorithm:sparse_cyclic_layout} that fully utilizes all $2\ell m$ ancilla qubits, in order to perform $X$ and $Z$ stabilizer measurements concurrently, leading to shorter circuit depths.

We begin with \cref{algorithm:IBMinterleavedBB6}, which is a modification of \cref{algorithm:sparse_cyclic_layout} to implement measurement of all stabilizer generators following a specific order.
Therein, $\mu$ is a set of 4-tuples $(u_z, u_x, \mathcal{Q}_z, \mathcal{Q}_x)$ that encapsulates a particular ordering of gates and cyclic shifts, where $u_z,u_x\in\mathbb{Z}_2$ while $Q_z$, $Q_x$ are either monomial constituents of $\polyA$, $\polyB$ or $-$ (indicating no associated operations for that step in $\mu$).

The overall depth of \cref{algorithm:IBMinterleavedBB6} is made significantly shorter with a modest generalization of our $2\times m$ model to $3\times m$.
In this generalization, each ancilla module $M_w^a$ is split into two modules each of size $\ell$; therein, the ancilla qubits are divided between the smaller modules as: $(X,v,w)\in M_w^{a_x}$ and $(Z,v,w)\in M_w^{a_z}$ for $v\in \mathbb{Z}_\ell$.
Furthermore, modules $M_w^{a_x}$ and $M_w^{a_z}$ are placed in cell $w$ of two distinct moving rows, and each moving row can undergo simultaneous and independent cyclic shifts.
Note that such a generalization makes it possible for the cyclic shifts of steps~4~and~10 of \cref{algorithm:IBMinterleavedBB6} to occur simultaneously.
Furthermore, if the sequence $\mu$ is appropriately chosen, gates of steps~6~and~12 may also occur concurrently.

\begin{algorithm}
\DontPrintSemicolon
\SetAlgoLined
\KwIn{A BB code, and an explicit ordering $\mu$ of syndrome gates.\;}
\KwOut{A circuit measuring all stabilizer generators of the input code.\;}
Prepare all ancilla qubits in the state $\ket +$.\;
\For{$u_z,u_x,\mathcal{Q}_z,\mathcal{Q}_x \in \mu$}{
    \If{$\mathcal{Q}_z$ is a monomial}{
    Apply the cyclic shift aligning $M_w^{a_z}$ and $\datamodule_{j}$.\;
    \For{$v, w \in G_{\ell, m}$}{
        Apply the CZ gate controlled on qubit $(Z,v,w)$ targeting qubit $(u_z, v \oplus h, w \oplus j)$.\;
    }}
    \ElseIf{$\mu$ is exhausted}{
    Measure \& reset all ancilla qubits $(Z,*,*)$.\;
    }
    \If{$\mathcal{Q}_x$ is a monomial}{
    Apply the cyclic shift aligning $M_w^{a_x}$ and $\datamodule_{k}$.\;
    \For{$v, w \in G_{\ell, m}$}{
    Apply the CX gate controlled on qubit $(X,v,w)$ targeting qubit $(u_x, v \oplus i, w \oplus j)$.\;
    }}
    \ElseIf{$\mu$ is exhausted}{
    Measure \& reset all ancilla qubits $(X,*,*)$.\;
    }
}
\caption{Sparse interleaved cyclic layout for BB codes.}
\label{algorithm:IBMinterleavedBB6}
\end{algorithm}

\begin{table*}
    \begin{centering}
    \begin{tabular}{|c|c|c|c|c|}
    \hline 
    \multirow{2}{*}{Layout Variant} & \multicolumn{3}{c|}{Depths for $T$ syndrome rounds} & \multirow{2}{*}{Amortized depth per round}\tabularnewline
    \cline{2-4} \cline{3-4} \cline{4-4} 
    & 2q Gates & Cyclic shifts  & Meas.+Reset & \tabularnewline
    \hline 
    \hline 
    Algo-\ref{algorithm:sparse_cyclic_layout}  & $2\omega T$  & $2T\left|J(\mathcal{A})\cup J(\mathcal{B})\right|$  & $4T$  & $2\left|J(\mathcal{A})\cup J(\mathcal{B})\right|+2\omega+4$\tabularnewline
    \hline 
    Algo-\ref{algorithm:flat_layout}  & $2\omega T$  & $4\omega T$  & $4T$  & $6\omega+4$\tabularnewline
    \hline 
    Algo-\ref{algorithm:IBMinterleavedBB6} (interleaved gates) & $\omega T+1$  & $\omega T+1$  & $2T$  & $2\omega+2$\tabularnewline
    \hline 
    Algo-\ref{algorithm:IBMinterleavedBB6} (concurrent rounds) & $\omega T+\omega$  & $T\left|J(\mathcal{A})\cup J(\mathcal{B})\right|+\left|J(\mathcal{A})\right|$  & $2T$  & $\left|J(\mathcal{A})\cup J(\mathcal{B})\right|+\omega+2$\tabularnewline
    \hline 
    \end{tabular}
    \par\end{centering}
    \caption{Table comparing variants of the sparse layouts of \cref{algorithm:sparse_cyclic_layout} and \cref{algorithm:IBMinterleavedBB6}.}
    \label{tab:DepthComparisons} 
\end{table*}

We define the {\em interleaved gates} layout to be an instance of \cref{algorithm:IBMinterleavedBB6}, with the gate-ordering of \cite{bravyi2024high}.
Writing polynomials of the BB code as $\polyA=\sum_{j=0}^2 A_j$ and $\polyB=\sum_{j=0}^2 B_j$, the interleaved gates layout is given by the following tuple sequence:
\begin{align*} 
    \mu = & \Big\{(1,-,A_{0}^{T},-),(1,0,A_{2}^{T},A_{1}),(0,1,B_{0}^{T},B_{1}),\\
    & (0,1,B_{1}^{T},B_{0}),(0,1,B_{2}^{T},B_{2}),(1,0,A_{1}^{T},A_{0}),\\
    & (-,0,-,A_{2})\Big\}
\end{align*}
It is accepted folklore that syndrome extraction circuits interleaving gates from $X$ and $Z$ stabilizer measurements in this way generally exhibit better logical error rates compared to circuits that implement $X$ and $Z$ stabilizer measurements non-concurrently.
Such a gate-ordering completes each BB code syndrome round in gate depth and cyclic-shift depth both $|\mu|-1 = \omega = 6$.
Gates associated with the last tuple can occur concurrently with those of the first tuple of a subsequent syndrome round, and the very first cyclic shift for the first tuple of $\mu$ is amortized over many syndrome rounds.

We also define the {\em concurrent rounds} layout as another instance of \cref{algorithm:IBMinterleavedBB6}, with the following tuple sequence $\mu=\mu_Z\cup \mu_{ZX}\cup \mu_X$ and:
\begin{align*}
    \mu_{Z} & =\bigcup_{A\in\mathcal{A}}\left\{ (1,-,A^{T},-)\right\} \\
    \mu_{ZX} & =\bigcup_{B\in\mathcal{B}}\left\{ (0,1,B^{T},B)\right\} \\
    \mu_{X} & =\bigcup_{A\in\mathcal{A}}\left\{ (-,0,-,A)\right\} 
\end{align*}
Observe that qubits with assigned actions in $\mu_Z$ and $\mu_X$ are non-overlapping.
Therefore ancillae qubits in $M^{a_z}$, which have no assigned action in $\mu_X$ during the trailing iterations of \cref{algorithm:IBMinterleavedBB6} for a current syndrome round, can be measured and reset to begin executing gates in $\mu_X$ for a subsequent syndrome round.

The ordering of operations in the concurrent rounds layout can be more flexible than that of the interleaved gates, since the ordering of monomials is entirely arbitrary in the construction of $\mu_Z$, $\mu_{ZX}$, and $\mu_X$.
For instance, the same ordering of monomials as in \cref{algorithm:sparse_cyclic_layout} can be chosen.
Except, in this concurrent rounds layout, only 1/2 as many cyclic shift steps is needed per syndrome round (with costs of executing operations of $\mu_Z$ in the very first round being amortized over many syndrome rounds).

\Cref{tab:DepthComparisons} summarizes variations of \cref{algorithm:sparse_cyclic_layout,algorithm:flat_layout,algorithm:IBMinterleavedBB6}.
We show depths incurred by 2-qubit gates, cyclic shifts, and measurement operations, \emph{disregarding} the particular physical constraints of \cref{sec:physical implementation} ({\em e.g.} on gate parallelism) as used in our numerical simulations.
Also shown is the overall depth per round, amortized over many syndrome rounds.
Notably, the {\em interleaved gates} and {\em concurrent rounds} layout of this section lower the circuit depth of \cref{algorithm:sparse_cyclic_layout,algorithm:flat_layout} by up to $2\times$.

\section{Fitting formulas for BB codes under the sparse cyclic layout and flat cyclic layout}
\label{appendix:fitting_formulas}

\begin{table}
    \centering
    \begin{tabular}{|c|c|c|c|}
    \hline
     $[[n,k,d]]$, layout & $c_0$ & $c_1$ & $c_2$   \\
     \hline
     $[[72, 12, 6]]$, sparse cyclic & 12.002 & 674.98 & -67694  \\
     \hline
     $[[90, 8, 10]]$, sparse cyclic & 24.397 & -290.59 & 24215  \\
     \hline
     $[[108, 8, 10]]$, sparse cyclic & 22.137 & 683.86 & -72746  \\
     \hline
     $[[144, 12, 12]]$, sparse cyclic & 28.049 & 375.30 & -42586  \\
     \hline
     $[[72, 12, 6]]$, flat cyclic & 11.963 & 408.55 & -29498  \\
     \hline
     $[[90, 8, 10]]$, flat cyclic & 24.105 & -325.04 & 34571  \\
     \hline
     $[[108, 8, 10]]$, flat cyclic & 21.678 & 522.45 & -43848  \\
     \hline
     $[[144, 12, 12]]$, flat cyclic & 27.422 & 140.49 & 3216.1  \\
     \hline
    \end{tabular}
    \caption{Constants in the fitting formula for the logical error rate of BB codes  $p_L=p^{d/2}e^{c_0+c_1p+c_2p^2}$ under the sparse cyclic layout in \cref{algorithm:sparse_cyclic_layout} and the flat cyclic layout in \cref{algorithm:flat_layout}.}
    \label{tab:fitting_formulas}
\end{table}

In this paper, by logical error rate we mean logical error rate per syndrome extraction round, not normalized by the number of logical qubits.
It is estimate using the same procedure as in \cite{ye2025quantum}.

Fitting formulas for logical error rates of surface codes and BB codes were studied under the circuit model with parallel gate operations and uniform noise rates \cite{fowler2012surface,bravyi2013simulation, bravyi2024high}.
For the ion chain model, the authors of \cite{ye2025quantum} also provided fitting formulas for surface codes and BB5 codes introduced in that paper.

Here we use the formula $p_L=p^{d/2}e^{c_0+c_1p+c_2p^2}$ to fit the logical error rate of BB codes under the sparse cyclic layout in \cref{algorithm:sparse_cyclic_layout} and the flat cyclic layout in \cref{algorithm:flat_layout}, where $d$ is the code distance, $p_L$ is the logical error rate, and $p$ is the physical error rate. The constants for the 4 BB code instances under the two different layouts are listed in \cref{tab:fitting_formulas}.

% \bibliography{references}


\end{document}